\documentclass[11pt,amssymb]{article}

\usepackage{amssymb}
\usepackage{epsfig}
\usepackage{fullpage}

\newcommand{\comment}[1]{}

\newtheorem{theorem}{Theorem}
\newtheorem{definition}{Definition}
\newtheorem{lemma}[theorem]{Lemma}

\newtheorem{corollary}{Corollary}

\def\squareforqed{\hbox{\rlap{$\sqcap$}$\sqcup$}}

\newtheorem{example}{Example}

\newcommand{\qed}{\vrule height6pt width4pt\medskip}

\newcommand{\xor}{\oplus}
\newcommand{\QED}{\hfill\qed}
\def\squareforqed{\hbox{\rlap{$\sqcap$}$\sqcup$}}
\def\qed{\ifmmode\squareforqed\else{\unskip\nobreak\hfil
\def\mod{\ {\rm mod}\ }
\penalty50\hskip1em\null\nobreak\hfil\squareforqed
\parfillskip=0pt\finalhyphendemerits=0\endgraf}\fi}

\newcommand{\F}{{\Bbb F}}
\newcommand{\N}{{\mathbb N}}

\begin{document}
\title{Polynomials: a New Tool for Length Reduction in Binary Discrete
Convolutions}

\author{
\begin{tabular}{cccc}
Amihood Amir\thanks{ Department of Computer Science, Bar-Ilan
University, Ramat-Gan 52900, Israel, +972 3 531-8770; {\tt
amir@cs.biu.ac.il}; and Department of Computer Science, Johns Hopkins
University, Baltimore, MD 21218. Partly supported by NSF grant
CCR-09-04581 and ISF grant 347/09, and BSF grant 2008217.}
&
Oren Kapah\thanks{Department of Computer Science,
Bar-Ilan U., 52900 Ramat-Gan, Israel, (972-3)531-8408, {\tt
kapaho@cs.biu.ac.il}.}
&
Ely Porat\thanks{ Department of Computer Science, Bar-Ilan
University, 52900 Ramat-Gan, Israel, (972-3)531-7620; {\tt
porately@cs.biu.ac.il}.}
&
Amir Rothschild\thanks{ Department of
Computer Science, Bar-Ilan University, 52900 Ramat-Gan, Israel,
(972-3)531-7620; {\tt amirrot@gmail.com}.}
\\
{\small Bar-Ilan University} & {\small Bar-Ilan University}
& {\small   Bar-Ilan University} & {\small Bar-Ilan University}\\
{\small and} \\
{\small Johns Hopkins University} \\
\end{tabular}
}

\date{}

\maketitle
\begin{abstract}
Efficient handling of sparse data is a key challenge in Computer Science.
Binary convolutions, such as polynomial multiplication or the Walsh
Transform are a useful tool in many applications and are efficiently solved.

In the last decade, several problems required efficient solution of
sparse binary convolutions. Both randomized and deterministic
algorithms were developed for efficiently computing the sparse
polynomial multiplication.
The key operation in all these algorithms was length reduction. The
sparse data is mapped into small vectors that preserve the convolution
result. The reduction method used to-date was the modulo function
since it preserves location (of the "1" bits) up to cyclic shift.

To date there is no known efficient algorithm for computing the sparse
Walsh Transform. Since the modulo function does not preserve the Walsh
transform a new method for length reduction is needed.
In this paper we present such a new method -– polynomials.
This method enables the development of an efficient algorithm for
computing the binary sparse Walsh Transform. To our knowledge, this is
the first such algorithm. We also show that this method allows a
faster deterministic computation of the sparse polynomial
multiplication than currently known in the literature.  

\end{abstract}

\section{Introduction\label{sec:intro}}

Handling sparse data is one of the grails of Computer Science, and is
a challenge in many of its sub-fields, from Machine Learning and Data
Bases to Combinatorics and Statistics. Even in one dimension sparse
data needs to be handled, and it is a tool for solving a number of
problems. An example is the {\em point-set matching problem}, where
two sets of points $T,P \in \N^d$ consisting of $n,m$ points,
respectively, are given. The goal is to determine if there is a rigid
transformation under which all the points in $P$ are covered with
points in $T$. Among the important application domains to which this
problem contributes are model based object
recognition~\cite{rucklidge:96,cyc:02}, pharmacophore
identification~\cite{kavraki:02}, searching in music archives,
~\cite{ulm:03}, and more.

The point-set matching problem has been studied in the literature in
many variation, not the least of which in the algorithms literature
(see e.g.~\cite{sc:98Sch,imv:99,CH:02,akof:03,acg:00,ahsw:03,ds:11}).
The template of many of these algorithms is as follows: First do a
{\it Dimension Reduction} where the inputs $T,P$ are linearized into
raw vectors $T',P'$ of size polynomial in the number of non-zero values.
reduced the $d$-Dimensional point set matching problem into a
$1$-Dimensional point set matching problem, and then solve the one
dimensional problem. This problem is also known as the {\em Sparse
Convolution Problem} 
We define it formally: 

\begin{definition}
Let $T:\{ 0,..., N-1\}\rightarrow \{0,1\}$ and $P:\{ 0,..., M-1
\}\rightarrow \{0,1 \}$ be binary functions. We may view $T$ and
$P$ as vectors of bits, whose lengths are $N$ and $M$, respectively.

{\sl INPUT:} Binary vectors $T$ and $P$ of length $N$ and $M$,
respectively.\\
{\sl OUTPUT:} All indices $i$ where, for all $j,\ \ j=0,...,M-1$,
either $T[i+j]=P[j]$ or $T[i+j=1$ and $P[j]=0$. We call such an index
  a {\em match}.
\end{definition}

{\bf Intuition:} If we consider a $1$ as depicting a point on the line
and a $0$ as having no point there, then the meaning of the problem is
to find all locations where every point of $P$ matches some point of $T$.

\begin{example}
$T=00000100101100011001010101110000000100$, and
$P=1000101$. There is a match in location $15$ because the situation
is:\\
\hskip 1in $T= ...\ 1\ 1\ 0\ 0\ 1\ 0\ 1\ 0\ ...$\\
\hskip 1in $P= ...\ 1\ 0\ 0\ 0\ 1\ 0 \ 1$\\
There are matches also at locations $19$ and $21$.
\end{example}

The above problem can be trivially solved in time $O(NM)$. It can be
solved in time $O(N\log M)$, in a computational model with word size
$O(\log M)$, using the Fast Fourier Transform (FFT)~\cite{CLR-92}.

Polynomial multiplication is a special case of a general convolution. 
A general convolution uses two initial functions, $v_1$ and $v_2$, to
produce a third function $w$. We consider a subclass of discrete
convolutions which we call {\em discrete dot-product convolutions}.

\begin{definition}\label{d:dot}
Let $v_1=v_1[0],...,v_1[N-1];v_2=v_2[0],...,v_2[N-1]$ be vectors in
$\mathbb{R}^N$. The {\em dot product} of $v_1$ and $v_2$, denoted as
$v_1\cdot v_s$, is: 
$$ v_1\cdot v_2 = \sum_{i=0}^{N-1} v_1[i]v_2[i].$$
\end{definition}

A discrete dot product convolution, operates $O(N)$ bijections on the
indices of $v_1$ and computes the dot product. Formally:

\begin{definition}\label{d:ddpt}
Let $v_1=v_1,v_2\in \mathbb{R}^N, c\in \mathbb{N}$. Let $B=\{
\beta_0,...,\beta_{cN-1} \}$ be bijections from $\{0,...,N-1$ to
$\{0,...,N-1\}$.
Then the {\em discrete dot-product B convolution of $v_1$ and $v_2$},
denoted as $v_1 \otimes_B v_2$, is the vector of length $cN$ defined as:
$$ v_1 \otimes_B v_2[j] = \sum_{i=0}^{N-1}
v_1[\beta(i)]v_2[i],\ \ j=0,...,cN-1.$$ 
\end{definition}

The polynomial multiplication we have seen above, as well as the
Discrete Walsh Transform (used e.g. in~\cite{aalp07} for computing
matching with flipped bits), are special cases of discrete dot-product
convolutions. Formally:

\begin{definition}\label{d:2dts}
\begin{enumerate}
\item The {\em Discrete Walsh Transform (DWT)} of $v_1,v_2\in
\mathbb{R}^N$ is the {\em discrete dot-product B convolution of $v_1$
and $v_2$} where $B=\{\beta_0,...,\beta_{N-1} \}$  and
$\beta_j:\{0,...,N-1\} \rightarrow \{0,...,N-1\}$ is defined as:
$\beta_j(i) = i \oplus j$. $i$ and $j$ are the binary number
representations of $i$ and $j$ (i.e. binary strings of length $\log N$)
and $\oplus$ is the {\em exclusive or} operation.
\item The {\em polynomial multiplication} of $v_1\in
\mathbb{R}^N, v_2\in \mathbb{R}^M$ is the {\em discrete dot-product B
convolution of $v_1$ and $v_2$} where $B=\{\beta_0,...,\beta_{N-1} \}$  and
$\beta_j:\{0,...,M-1\} \rightarrow \{j,...,j+M-1\}$ is defined as:
$\beta_j(i) = i + j$. In this definition the dot product is of vectors
  in $\mathbb{R}^M$, i.e. $$ v_1 \otimes v_2[j] = \sum_{i=0}^{M-1}
v_1[i+j]v_2[i],\ \ j=0,...,N-M+1.$$  
\end{enumerate}
\end{definition}
%
%

A special feature both these convolutions have is that they can be
solved in time $O(N\log N)$ via a divide-and-conquer approach. The
algorithm for the polynomial multiplication uses the FFT and the
algorithm for the Discrete Walsh Transform is the {\em Fast Walsh
Transform (FWT)}. Of course, similar to the DFT case, one may also 
consider a binary version of the DWT. In both
these cases, the $O(N\log N)$ divide-and-conquer algorithm is a
wonderful solution if the input vectors have many points. Suppose,
however, that our arrays are {\em sparse}. In the sparse case, many
values of $T$ and $P$ are $0$. Thus, they do not contribute to the
convolution value. 

{\bf Convention:} Throughout this paper, a capital letter (e.g. $N$)
is used to denote the size of the vector, which is equivalent to the
largest index of a non-zero value, and a small letter (e.g. $n$) is
used  to denote the number of non-zero values. It is assumed that the
vectors are not given explicitly, rather they are given as a set of
$(index,value)$ pairs, for all the non-zero values.

In our convention, the number of non-zero values of $T (P)$ is $n
(m)$. Clearly, we can compute the convolution either in time $O(N\log
M)$ or in time $O(nm)$. The challenge was (see e.g.~\cite{muthu-open})
whether the convolution can be computed in time $o(nm)$.

To our knowledge, the only efficient algorithms for computing sparse
discrete binary convolutions, are for the FFT. The state-of-the art in
computing such sparse convolutions is to use a locality preserving
function to {\em reduce the length} of the sparse vectors, and then
to use the fast convolution algorithm. In Section~\ref{s:pre} we 
summarize the details of current length reduction methods. The
locality preserving function is the {\sl modulo} function.

{\bf Paper's Contribution:} The main contribution of this paper is
introducing a {\em novel tool for length reduction} -- {\sl
  polynomials}.
We show a number of advantages to using polynomials:
\begin{enumerate}
\item The new polynomials technique leads to an elegant algorithm for
  the sparse DWT. To our knowledge, no such algorithm is known to-date.
\item Our technique can be used without preprocessing to achieve a
  Las-Vegas algorithm for the sparse FFT which runs in time $O(n\log^2
  n)$. This matches the expected running time of the Las-Vegas
  algorithm presented by Cole and Hariharan in~\cite{CH:02}. However,
  our algorithm has the added advantage of guaranteeing a worst case
  time of $O(n^2\log n)$ whereas the algorithm in~\cite{CH:02} has no
  bound on its worst case running time.
\item Allowing $O(n^2)$ preprocessing time on the text, our technique
  can then {\em deterministically} solve the sparse FFT for incoming
  patterns in time $O(n \log^2 n)$. This improves the best known
  deterministic solution to the problem (\cite{LR07}), whose time was
  $O(n \log^3 n)$, with the same preprocessing time.
\end{enumerate}

 \section{Summary of Length Reduction Methods}\label{s:pre}

Recently~\cite{hikpa:12,hikpb:12}, Hassanjeh et
al. considered algorithms for approximating the sparse Fourier
transform. We are interested in exactly computing sparse convolutions 
(the Walsh transform, as well as the fast Fourier transform-based
polynomial multiplication). In this section we summarize Cole and
Hariharan's {\em length reduction} method, which exactly computes
sparse polynomial multiplication. in order to pinpoint the role of the
locality preserving function. 

Suppose we can map all the non-zero values into a smaller vector, say
of size $O(n \log m)$. Suppose also that this mapping is alignment
preserving in the sense that applying the same transformation on $P$
will guarantee that the alignments are preserved. Then we can simply
map the the vectors $T$ and $P$ into the smaller vectors and then
use FFT for the convolutions on the smaller vectors, achieving time
$O(n \log^2 m)$. We then map the results back to the original vectors.

The problem is that to-date there is no known mapping with that
alignment preserving property. Cole and Hariharan~\cite{CH:02}
suggested using the modulo function as the alignment preserving
mapping, in a randomized setting that answers the problem with high
probability. The reason their algorithm is not deterministic is the
following:
In their length reduction phase, several indices of
non-zero values in the original vector may be mapped into the same
index in the reduced size vector. If the index of only one non-zero
value is mapped into an index in the reduced size vector, then this
index is denoted as {\it singleton} and the non-zero value is said
to appear as a {\it singleton}. If more then one non-zero value is
mapped into the same index in the reduced size vector, then this
index is denoted as {\it multiple}. The multiple case is problematic
since we can not be sure of the right alignment. The proposed solution
of Cole and Hariharan was to create a set of $\log n$ pairs of vectors
using $\log n$ hash function rather then a single pair of vectors.
They showed that in $O(\log n)$ attempts, the probability that some
index will {\em always} be in a multiple situation is small.

A different, deterministic, solution was shown in~\cite{LR07}. The
idea was to find $\log n$ hash functions that reduce the size of the
vectors to $O(n \log n)$. The algorithm guaranteed that each
non-zero value appears with no collisions in {\em at least} one of the
vectors, thus eliminating the possibility of en error.

The ultimate goal of the length reduction is as follows: Given two
vectors $T,P$ whose sizes are $N,M$, with $n,m$ non-zero
elements respectively (where $n>m)$, obtain two vectors
$T',P'$ of size $O(n)$ such that all the non-zero elements in $T$ and 
in $P$ will appear as singletons in $T'$ and in $P'$
respectively while maintaining the distance property.

The {\em distance property} which need to be maintained is defined as
follows: If $P'[f(0)]$ is aligned with $T'[f(i)]$, then
$P'[f(j)]$ will be aligned with $T'[f(i+j)]$.

This goal was not reached yet, rather a set of $O(\log n)$ vectors
of size $O(n \log n)$ were obtained in ~\cite{LR07}, where each
non-zero in the text appears at least once as a singleton in the set
of vectors. This length reduction gave an $O(n \log^3 n)$
algorithm for convolution in sparse data with a preprocessing time of
$O(n^2)$.

The general outline of this idea could, conceivably, work for the
sparse FWT case. However, the modulo function can not serve as a
locality preserving mapping for the Walsh transform, since the
convolution bijection is an exclusive or rather than a shift. Thus
to-date there is no known efficient algorithm for the sparse DWT.

\section{The New Length Reduction Technique for the FWT}\label{s:new_fwt}

We employ a length reduction algorithm that applies to any
convolution. The algorithm has four stages:
\begin{enumerate}
\item {\em Reduction}: A locality preserving function that reduces the
  range of indices.
\item {\em Convolution}: Perform the convolution on the shorter
  strings.
\item {\em Verification}: Verify that the results are consistent,
  i.e. that all multiplied elements indeed needed to be
  multiplied. 
\item {\em Solution Expansion}: Map the solution of the shorter
  strings, to the original longer strings.
\end{enumerate}

Below we explain the intuition for the reduction function chosen for
the DWT, and then give the details of how polynomials perform the
required length reduction.

Recall that in the case of FWT both the text and pattern have the same
length $N$. For ease of exposition we assume that $N=2^L$. Assume that
the text has $n$ non-zero values, and the pattern has $m$ non-zeros
values. Recall also that we denote the {\em exclusive or} operation by
$\xor$. As in the sparse FFT case, we want to check the value of the
FWT for all $i$ for which  $\sum_{j=0}^{n-1} T[i \xor j]P[j]$ has $m$
non-zero summands. As we mentioned above, the algorithm is the same as in the
case of the DFT. The problem is that here the 'modulo' function is no
help at all as a length reduction function. 

\subsection{Intuition}\label{ss:int_fwt}
Split the text to two strings of size
$2^{L-1}=\frac{N}{2}$ and merge them with the ``or'' operator. Due to
the fact that $N>>n$, with high probability there will be no
collisions, and thus all cases will be singletons. We do the same
thing to the pattern and then run the FWT on the smaller strings. The
multiples are handled similarly to the FFT case, i.e. ignore at this
stage and in another reduction they will be unlikely to collide. We
also need a verification stage to make sure that we have not
multiplied values from different substrings.

This argument works well if the inputs are generated under a uniform
random distribution. For supporting all inputs we need some
randomization tool. The simple way to do it is as follow: Let $T_1$
be the first half of $T$ and $T_2$ the second. Choose a random bit
string $mask$ of length ${L-1}$, set its most significant bit to be
$1$, and calculate $T_2^{mask}[i]=T_2[i \xor mask]$. Now do the `or'
operation between $T_1$ and $T_2^{mask}$. In effect, we have randomly
permuted $T_2$. The following lemma, which immediately follows from
the commutativity of the exclusive-or operation,  guarantees that this 
random permutation preserves the locality for the Walsh transform.

\begin{lemma}\label{l:local_walsh}
Pattern index $i$ matches text index $k$ for location $j$ of the
Walsh transform iff pattern index $i\xor mask$ matches text index $k
\xor mask$.
\end{lemma}

\begin{example}\label{ex:1} Consider the following text and pattern:

\begin{tabular}{|l|c|c|c|c|c|c|c|c|}
\hline
address & 000 & 001 & 010 & 011 & 100 & 101 & 110 & 111\\
\hline\hline
T & 0 & 1 & 0 & 0 & 0 & 0 & 1 & 0\\
\hline
P & 1 & 0 & 0 & 0 & 0 & 0 & 0 & 1\\
\hline
\end{tabular}

The Walsh transform is: $02000020$ (since at locations $001=1$ and
$110=6$ the $1$'s in the text and pattern match).

Consider now $mask = 101$. (The most significant bit (MSB) is always
$1$ at the mask because it is a mask on $T_2$, which is the second
half of the text). The smaller text and pattern are as below:

\begin{tabular}{|l|c|c|c|c|}
\hline
address & 000 & 001 & 010 & 011 \\
\hline\hline
T' & 0 & 1 & 0 & 1\\
\hline
P' & 1 & 0 & 1 & 0\\
\hline
\end{tabular}

The Walsh transform is $0202$. 

\end{example}

The above example lacks the verification and solution expansion
stages. We need to verify that the results are {\em
  consistent}. Consistency means that one of two cases happen:
\begin{enumerate}
\item Every text value with an unchanged address is multiplied by a
pattern value with an unchanged address, and every text value with a
masked address is multiplied by a pattern value with a masked address. 
\item Every text item with an unchanged address is multiplied by a
  pattern item with a masked address and every text item with a masked
  address is multiplied by a pattern item with an unchanged address.
\end{enumerate}

In addition it is necessary to be able to identify which of the non-zero
results belong to their location ($01$ in the example) and which should
be xor-ed with the mask ($11$ in the example).

{\bf Verification Idea:} In the reduction stage, construct two strings
of length $2^{L-1}$ where for each non-zero value write $s$ if it
is static, i.e., was not moved by the mask, and $m$ if it was moved by
the mask. We are interested in products where all values are either $s
\cdot s$, or $m \cdot m$, and in products where all values are
$m\cdot s$. This can be calculated by a constant number of
appropriate binary DWTs.

\begin{example}\label{ex:2}
In Example~\ref{ex:1} we have

\begin{tabular}{|l|c|c|c|c|}
\hline
address & 000 & 001 & 010 & 011 \\
\hline\hline
T' & 0 & s & 0 & m\\
\hline
P' & s & 0 & m & 0\\
\hline
\end{tabular}

Locations $00$ and $10$ have value $0$.
In location $01$ we have a product of $m\cdot m$ and $s\cdot s$, and
in location $11$ we have products of $m\cdot s$. Thus the reduction
gives a consistent result and, furthermore, it means that the result
in location $01$ remains in that location, and the result in location
$11$ belongs to the index xor-ed with the mask, i.e. $11 \xor 101
=110 = 6$.
 
\end{example}

The idea above reduces the length by a half, from $N$ to
$\frac{N}{2}$. It is easy to see that this process can be recursively
repeated. However, the analysis of the probability of collision is
then a bit tedious. Our polynomial framework for length reduction
enables this analysis in a clear and elegant manner.

\subsection{Polynomials over a Finite Field}\label{ss:poly_fwt}

We will consider our indices as elements in the finite field
$F_{2^L}$. The motivation for this is that $x+y=x \xor y$ in $F_{2^L}$.
Now in order to reduce the length of the string we will write each
element in $F_{2^L}$ as a polynomial in $F_{2^{\ell}}[X]$ of degree
$d=\frac{L}{\ell}-1$. In order to do that, let us look at the binary
presentation of a non-zero element index $i=a_{L-1}a_{L-2}...a_0$. The
index binary presentation is divided into set of $\ell$ bits starting
from the LSB. By this division we obtain a set of $\frac{L}{\ell}$
numbers in $F_{2^{\ell}}$. These numbers will be used as the
coefficients of the polynomial representing this index, thus this
polynomial belongs to $F_{2^{\ell}}[X]$ and its degree is
$d=\frac{L}{\ell}-1$. 

\begin{example} 

Consider the non-zero element index $i=17$. let us assume that
$\ell=2$ by breaking the binary presentation of $i$ into blocks of $2$
bits we obtain $i=17=10001=(01)*X^2+(00)*X+(01)*1=X^2+1$. 

\end{example}

Note that under $F_{2^{\ell}}[X]$, addition is calculated as XOR, and
multiplication is calculated as polynomial multiplication modulo some
irreducible polynomial. Below is an example of the multiplication
table of $F_{2^2}$ 

\begin{example} Multiplication table of $F_{2^2}$. 

\begin{tabular}{|l|c|c|c|c|}
\hline
 & 0 & 1 & a & b \\
\hline\hline
0 & 0 & 0 & 0 & 0\\
\hline
1 & 0 & 1 & a & b\\
\hline
a & 0 & a & b & 1\\
\hline
b & 0 & b & 1 & a\\
\hline
\end{tabular}

Note that $a$ and $b$ can be thought as representing $10$ and $11$ respectively.
\end{example}

After we obtain the polynomials representing
the indices of the non-zeros elements, we can obtain different reduced
size vector by applying different values to $X$ in the
polynomials. 

To prove the correctness of this length reduction for FWT, we have to
prove that the following lemma still holds: 

\begin{lemma}\label{l:WalAllignment}
For any assignment of $X$, if $P[0]$ is aligned with the base
polynomial representing $T[i]$, then $P[j]$ will be aligned with
one of the polynomials representing $T[i \oplus j]$.
\end{lemma}

\begin{proof}

The proof is simple, and it is based on the building method. The
polynomial representing the index of $P[0]$ is $P_0[X]=0$ and it is
aligned with $P_i[X]$, which is the polynomial representing the index
$T[i]$, thus the polynomial representing $P[j]$ is $P_j[X]$, and it is
aligned with the polynomial $P_i[x]+P_j[x]$. Recall that the
coefficients of the polynomials $P_i[x]$ and $P_j[x]$ are the bits
representing $i$ and $j$, and the addition under $F_{2^\ell}$ is XOR,
thus we get that $P_i[x]+P_j[x]$ is the polynomial representing the
index $i \oplus j$. 

\end{proof}

In this framework, the analysis of the probability of collision is
easy. Two indices $i=p_i(X)$ and $j=p_j(X)$ can collide only on
$d=\frac{L}{\ell}-1$ different evaluations, as any collision implies
that the value assigned to $X$ in the polynomials is a root of the
difference polynomial representing those two indices. Since the degree
of this polynomial is bounded by $d$, then there can be no more then
$d$ different roots to this polynomial. We conclude: 

\begin{lemma}\label{l:prob_collision}
The probability of a collision when choosing a random evaluation $r$
is no greater than $\frac{d}{2^{\ell}}$.
\end{lemma}

It is easy to see that appropriate binary Walsh transforms allow us to
detect singletons and delete them from the next rounds. Appropriate
binary Walsh transforms also easily provide the expansion.

\section{The New Length Reduction Technique for the FFT}\label{s:new_fft}

\subsection{Sparse Vector of Polynomial Length}\label{s:length}

The proposed technique deals with the case that $N$ is polynomial
in $n$, thus the indices are bounded by $n^c$. In the case
where, $N$ is exponential in $n$, the reduction to a polynomial
case can be used.

{\bf Main Idea:} Derive a set of unique polynomials
from each non-zero index in $T$, and one polynomial for each non-zero
in $P$. Each assignment for the polynomials in $\F_q$, where $q$ is a
prime number of size $\Theta (n)$ will give a different mapping of
the non-zeros in $T$ and in $P$ to vectors of size $q$. The
convolution will be performed between the vectors obtained from
$T$ and $P$ under the same assignments.

The first step of the algorithm is to choose a prime number of size
$\Theta (n)$, and create a polynomial for each non-zero index in
$T$. The created polynomial of index $i$ will be denoted as the {\em
base polynomial} of $T[i]$. The creation of the polynomial is done by
representing the index as a number in base $(q-1) \over 2$. Each
digit is interpreted as a coefficient of the polynomial.

\begin{example}
Let $q=13$. Consider index $95$ in base $10$. This equals $235$ in
base ${(13-1) \over 2}=6$. Each digit is a coefficient, thus the
polynomial representing index $95$ is $2X^2+3X+5$.
\end{example}

Since the indices in $T$ are bounded by $n^c$, and $q$ is $\Theta
(n)$, then the degree of the polynomials which created in this step
is bounded by $c$. In the next step, from each polynomial we create
$2^c$ polynomials. This is done by giving two choices for each
coefficient of the polynomial:
(1) Leave it as is. (2) Add $(q-1) \over 2$ to the coefficient and
decrease by 1 the coefficient of the higher degree. We do this for all
the coefficients of the polynomial except for the coefficient of the
highest degree.

\begin{example}
Suppose we have a non-zero index $95$, using $q=13$ we get the base
polynomial $2X^2+3X+5$. After the second step we will obtain $4$
polynomials: $2X^2+3X+5$, $2X^2+2X+11$,$X^2+9X+5$,$X^2+8X+11$.\\
The first polynomial is the base polynomial. The second polynomial was
obtained by adding $6$ to the first coefficient and decreasing the
second coefficient by one. The 3rd and the 4th polynomials were created
by adding $6$ to the second coefficient of the first and second
polynomials respectively, and decreasing the third coefficient by one.
\end{example}

The duplication of the polynomials was made to meet the distance
preserving requirement of the length reduction. The following lemma
formalizes it.

\begin{lemma}\label{l:PolAllignment}
For any assignment of $X$, if $P[0]$ is aligned with the base
polynomial representing $T[i]$, then $P[j]$ will be aligned with
one of the polynomials representing $T[i+j]$.
\end{lemma}
\begin{proof}
Let $q$ be the chosen prime number. Index $0$ in $P$ is
represented by the polynomial $0$, and index $j$ in $P$ is
represented by the a polynomial $A=a_cX^c+a_{c-1}X^{c-1}+...+a_0$.
Index $i$ in $T$ is represented by a polynomial of the form
$B=b_cX^c+b_{c-1}X^{c-1}+...+b_0$, and index $i+j$ in $T$ is
represented by a polynomial $D=d_cX^c+d_{c-1}X^{c-1}+...+d_0$. Note
that the coefficients $a_i$ and $b_i$ are smaller then $(q-1) \over
2$.

Clearly, if $P[0]$ is aligned with $T[i]$, then for any
assignment of $X$, $P[j]$ will be aligned with the polynomial
$A+B=(a_c+b_c)X^c+(a_{c-1}+b_{c-1})X^{c-1}+...+(a_0+b_0)$. Now consider
the first coefficient of $D$. Since $a_0$ and $b_0$ are
smaller then $(q-1) \over 2$, then there are only two cases: (1)
$(a_0+b_0)<{(q-1) \over 2}$, thus $d_0=a_0+b_0$. (2)
$(a_0+b_0)>={(q-1) \over 2}$, thus $d_0=a_0+b_0-{(q-1) \over 2}$
which is covered by the
polynomial where $(q-1) \over 2$ was added to the first coefficient.\\
In the later case, one was added to the second coefficient, thus we
decrease the next coefficient whenever we add $(q-1) \over 2$ to the
current coefficient. The same cases exist also in all the
coefficients, but a polynomial was created for each possible case ($2^c$
cases), thus one of the created polynomials will be equal to the
polynomial $A+B$. $\QED$
\end{proof}

Note that all the $2^c \times n$ created polynomials are unique, and
in $\F_q$. Assigning a value to the polynomials in $\F_q$ will give a
vector of size $q$.

\begin{lemma}\label{l:maxRoots}
Any two polynomials can be mapped to the same location in at most
$c$ assignments.
\end{lemma}
\begin{proof}
The distance between any two polynomials gives a polynomial, where the
degree of the difference polynomial is bounded by $c$. Since both
polynomials give the same index under the selected assignment, then the
assigned value is a root of the difference polynomial. The degree of
this polynomial is bounded by $c$, thus it can have at most $c$
different roots in $\F_q$. $\QED$
\end{proof}

Since any polynomial can be mapped into the same location with at most
$2^c \times n-1$ other polynomials, and with each of them at most $c$
times, due to Lemma \ref{l:maxRoots}, then we get the following
Corollary:
\begin{corollary}\label{c:maxPolMultiples}
Any polynomial can appear as a {\it multiple} in not more then $c
\times 2^c \times n$ vectors.
\end{corollary}

The last step of the length reduction algorithm is to find a set of
$O(\log n)$ assignments which will ensure that each polynomial will
appear as a singleton at least once.

The selection of the $O(\log n)$ assignments is done as follows:
Construct table $A$ with $2^c \times n$ columns and $c \times
2^{c+1} \times n$ rows. Row $i$ correspond to an assigned value
$a_i$ and the corresponding reduced length vector $T'_{i}$. A
column corresponds to a polynomial $P_j$. The value of $A_{ij}$ is set
to $1$ if polynomial $j$ appears as a {\it singleton} in vector
$T'_{i}$. Due to Corollary \ref{c:maxPolMultiples}, the number of
zeros in each column can not exceed $c \times 2^c \times n$. Thus,
in each column there are $1$'s in at least half of the rows, which
means that the table is at least half full. Since the table is at
least half full there exists a row in which there is one in at least
half of the columns. The assignment value which generated this row
is chosen, and all the columns where there was a $1$ in the selected
row are deleted from the table.

Recursively another assignment value is chosen and the table size is
halved again, until all the columns are deleted. Since at each step
at least half of the columns are deleted, the number of prime number
chosen can not exceed  $\log (2^c \times n) = c \log n$.

{\bf Time:} Creating vector $T'_{i}$ (row $i$) takes $O(n)$ time.
Since we start with a full matrix of $O(n)$ rows then the
initialization takes $O(n^2)$ time. Choosing the $O(\log n)$
assignment values is done recursively. The recurrence is:
$$ t(n^2) = n^2 + t({n^2\over 2})$$
The closed form of this recurrence is $O(n^2)$.

Note that if we choose the assignments randomly in running time we can
skip the preprocessing step. In the expected case we will succeed
after $O(\log n)$ assignments. But even in the worst case $O(n)$
assignments will guarantee success.

\subsection{Sparse Vector of Exponential Length}\label{s:Exp}

In this case, as proposed in ~\cite{LR07}, each of the vectors $T$
and $P$ is reduced into a single vector of size $O(n^4)$, where
all the non-zeros appear as singletons. The reduction is performed
using the modulus function with a prime number $q$ of size
$O(n^4)$. It was already proven there that there are at most
$n^3$ prime number of size $O(n^4)$, which generate at least one
multiple. Thus, by testing $n^3+1$ prime numbers we ensure that at
least one of them produces a vector with no multiples.

In order to find such a prime number, we find $n^3+1$ prime numbers of
size $O(n^4)$. Then we multiply all the prime numbers to receive a
large number $Q$. In addition we have at most $n^2$ different
distances between any two non-zeros. We multiply all of them to
receive a large number $D$. The next step is to find the greatest
common divisor ($GCD$) between $Q$ and $D$. Since there is at least
one prime number in $Q$ which does not divide $D$, then $GCD(Q,D)$
is less then $Q$. Dividing $Q$ by the $GCD(Q,D)$ will give $R$ which
is the multiplication of all the prime numbers that create only
singletons. The last step is to find at least one of them. This is
done using a binary search on the prime numbers. We take the
multiplication of half of the prime numbers $Q'$, and find the
$GCD(Q',R)$. If $GCD(Q',R) > 1$ we continue with this set of prime
numbers and multiply half of them iteratively. Otherwise, we
continue with the other half of the prime numbers. After $O(\log
n)$ iterations we will find one prime number which will generate
only singletons.

The algorithm appears in detail below.

\fbox{
\begin{minipage}{16cm}
{\bf Algorithm  -- $N$ is exponential in $n$} {\sf
\begin{enumerate}
    \item Find $n^3+1$ prime numbers of size $O(n^4)$.
    \item Multiply all the prime numbers to obtain $Q$.
    \item Multiply all the difference between any two non-zero indices
      to obtain $D$.
    \item Set $R={Q \over GCD(Q,D)}$.
    \item Let $S$ be the set of all prime numbers.
    \item While the size of $S$ is larger then $1$ do:
\begin{enumerate}
    \item Let $S'$ be a set of the first half of prime numbers in $S$.
    \item Set $Q'$ to be the multiplication of all the prime numbers in $S'$.
    \item If $GCD(Q',R)>1$ then set $S=S'$, otherwise set $S=S/S'$.
\end{enumerate}
\end{enumerate}
{\bf end Algorithm} }
\end{minipage}
}

{\bf Correctness:} Immediately follows from the discussion.

{\bf Time:} Step 1 is performed in time $O(n^3 {\rm polylog}(n))$
using the primality testing described in ~\cite{berri:02}. Step 2 is
done by building a binary tree of products where each node
contain the product of the two number in the lower level. This
tree has $O(\log n)$ levels. In the leaves there are $n^3$ prime
numbers with $\log n$ bits, so the total number of bits in each
level is $O(n^3 \log n)$. A product of two numbers can be
computed in time $O(b \log b \log \log b)$ ~\cite{SS-71}, where $b$ is
the number of bits. Thus each level can be computed in time $O(n^3
{\rm polylog} (n))$ and the total time for step 2 is $O(n^3
{\rm polylog}(n))$. step 3 is preformed in the same way, but this time
in the leaves there are $n^2$ numbers with $n$ bits, thus each
level has $n^3$ bits and the time for this step is $O(n^3 \log
n)$. In step 4 we calculate the $GCD$ of two numbers with $O(n^3
\log n)$ bits. This can be done in time $O(n^3 {\rm polylog}(n))$
using ~\cite{SZ:02}. The calculation for step 6(b) was
already performed in step 2, and step 6(c) can be calculated in time
$O(n^3 {\rm polylog}(n))$, thus the time of step 6 is $O(n^3
{\rm polylog}(n))$. Following this discussion the total time of this
algorithm is $O(n^3 {\rm polylog}(n))$.

\section{Conclusion and Open Problems}\label{s:conc}

A new tool for Length Reduction and Sparse Convolution was introduced
in this paper - encoding the indices as polynomials. This enabled the
first algorithm for efficient calculation of sparse binary discrete
Walsh transform. It also provided better and faster algorithms for
several well known problems, such as randomized efficient sparse FFT
and deterministic efficient sparse FFT computation.

An important problem remains: Can the Length Reduction and Sparse
DFT problems be solved deterministically without the need of the
preprocessing step or, alternately, can the preprocessing time be
reduced from quadratic?

\bibliographystyle{plain}
\small{
\bibliography{/home/amir/Tex/paper}
}

\end{document}